\documentclass[runningheads,a4paper, 11pt]{article}

\usepackage[ansinew]{inputenc}
\usepackage[dvips]{graphicx}
\usepackage[english]{babel}
\usepackage{amsthm,amssymb,amsfonts,amsmath}
\usepackage {color}
\usepackage{hyperref}
\usepackage{xspace}
\usepackage{comment}
\usepackage{wrapfig}
\usepackage{algorithmic}
\usepackage{algorithm}
\usepackage[margin=3cm]{geometry}

\newcommand{\etal}{\emph{et al.}\xspace}

\newcommand{\ellminus}{\ensuremath{\ell^-}}
\newcommand{\ellplus}{\ensuremath{\ell^+}}

\newcommand{\Lone}[2]{\ensuremath{d_{L_1}(#1, #2)}}

\newcommand{\xdist}[2]{\ensuremath{|#1 #2|_x}}
\newcommand{\ydist}[2]{\ensuremath{|#1 #2|_y}}

\newcommand{\sq}[2]{\ensuremath{S(#2)}}

\newcommand{\mpath}[2]{\ensuremath{\sigma_{#1\to#2}}}

\newcommand{\n}[2]{\ensuremath{#2_{#1}}}

\newcommand{\cLone}{\ensuremath{c^*}}

\newtheorem{problem}{Problem}

\newtheorem{lemma}[problem]{Lemma}
\newtheorem{theorem}[problem]{Theorem}
\newtheorem{observation}[problem]{Observation}

\title{On the stretch factor of the Theta-4 graph}
\date{}
\author{Luis Barba\thanks{Carleton University, Ottawa, Canada} \thanks{Universit\'e Libre de Bruxelles, Brussels, Belgium} \and Prosenjit Bose$^*$ \and Jean-Lou De Carufel$^*$ \and Andr\'e van Renssen$^*$ \and Sander Verdonschot$^*$}

\begin{document}

\maketitle

\begin{abstract}
In this paper we show that the $\theta$-graph with 4 cones has constant stretch factor, i.e., 
there is a path between any pair of vertices in this graph whose length is at most a constant times the Euclidean distance between that pair of vertices. 
This is the last $\theta$-graph for which it was not known whether its stretch factor was bounded.
\end{abstract}

\section{Introduction}

A $t$-spanner of a weighted graph $G$ is a connected sub-graph $H$ with the property that for all pairs of vertices $u$ and $v$, the weight of the shortest path between $u$ and $v$ in $H$ is at most $t$ times the weight of the shortest path between $u$ and $v$ in $G$, for some fixed constant $t\geq 1$. The smallest constant $t$ for which $H$ is a $t$-spanner of $G$ is referred to as the \emph{stretch factor} or \emph{spanning ratio} of the graph. The graph $G$ is referred to as the \emph{underlying graph}. In our setting, the underlying graph is the complete graph on a set of $n$ points in the plane and the weight of an edge is the Euclidean distance between its endpoints. A spanner of such a graph is called a \emph{geometric spanner}. For a comprehensive overview of geometric spanners, see the book by Narasimhan and Smid \cite{NS06}.

\begin{wrapfigure}[16]{r}{0.39\textwidth}
  \begin{center}
  \vspace{-.4in}
\includegraphics{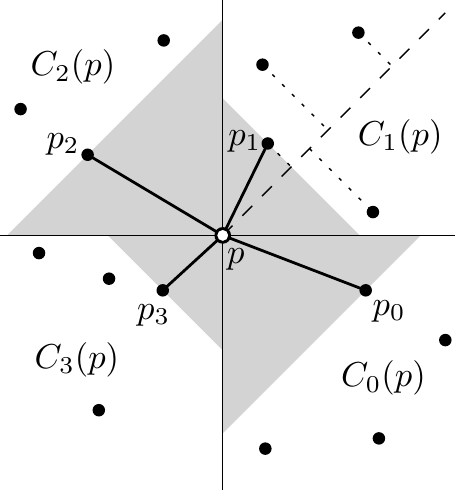}
\vspace{-.2in}
\caption{\small The neighbors of $p$ in the $\theta_4$-graph of $P$. Each edge supports an empty isosceles triangle.}
  \label{fig:Construction of Theta 4}
    \end{center}
\end{wrapfigure}

In this paper, we focus on $\theta$-graphs. Introduced independently by Clarkson~\cite{Cl87} and Keil~\cite{Keil88}, the $\theta_m$-graph is constructed as follows. Given a set $P$ of points in the plane, we consider each point $p \in P$ and partition the plane into $m$ cones (regions in the plane between two rays originating from the same point) with apex $p$, each defined by two rays at consecutive multiples of $\theta = 2\pi/m$ radians from the negative $y$-axis. We label the cones $C_0(p)$ through $C_{m-1}(p)$, in counter-clockwise order around $p$, starting from the negative y-axis; see~Fig.~\ref{fig:Construction of Theta 4}. In each cone $C_i(p)$, we add an edge between $p$ and $p_i$, the point in $C_i(p)$ nearest to $p$. However, instead of using the Euclidean distance, we measure distance in $C_i(p)$ by projecting each vertex onto the angle bisector of this cone.
Formally, $\n{i}{p}$ is the point in $C_i(p)$ such that for every other point $w\in C_i(p)$, 
the projection of $\n{i}{p}$ onto the angle bisector of $C_i(p)$ lies closer to $p$ than that of $w$.
For simplicity, we assume that no two points of $P$ lie on a line parallel to either the boundary or the angle bisector of a cone.

Ruppert and Seidel~\cite{RS91} showed that $\theta_m$-graphs are geometric spanners for $m \geq 7$, and their stretch factor approaches 1 as $m$ goes to infinity. Their proof crucially relies on the fact that, given two vertices $p$ and $q$ such that $q \in C_i(p)$, the distance between $p_i$ and $q$ is always less than the distance between $p$ and $q$. This property does not hold for $m \leq 6$ and indeed, the path obtained by starting at $p$ and repeatedly following the edge in the cone that contains $q$, is not necessarily a spanning path.
The main motivation for using spanners is usually to reduce the number of edges in the graph without increasing the length of shortest paths too much. Thus, $\theta$-graphs with fewer cones are more interesting in practice, as they have fewer edges. This raises the following question: ``What is the smallest $m$ for which the $\theta_m$-graph is a geometric spanner?'' Bonichon~\etal~\cite{BGHI10} showed that the $\theta_6$-graph is also a geometric spanner. Recently, Bose~\etal~\cite{bose2013theta5} proved the same for the $\theta_5$-graph. Coming from the other side, El~Molla~\cite{molla2009yao} showed that there is no constant $t$ for which the $\theta_2$- and $\theta_3$-graphs are geometric spanners. This leaves the $\theta_4$-graph as the only open question. Moreover, its resemblance to graphs like the $\textrm{Yao}_4$-graph~\cite{bose2012pi} and the $L_\infty$-Delaunay triangulation~\cite{StretchFactorL1Linfty}, both of which are spanners, make this question more tantalizing.
In this paper we establish an upper bound of approximately $237$ on the stretch factor of the $\theta_4$-graph, thereby showing that it is a geometric spanner.
In Section~\ref{sec:LowerBound}, we present a lower bound of $7$ that we believe is closer to the true stretch factor of the $\theta_4$-graph. 

\section{Existence of a spanning path}
Let $P$ be a set of points in the plane.
In this section, we prove that the $\theta_4$-graph of $P$ is a spanner. We do this by showing that the $\theta_4$-graph approximates the $L_\infty$-Delaunay triangulation.
The $L_\infty$-Delaunay triangulation of $P$ is a geometric graph with vertex set $P$, and an edge between two points of $P$ whenever there exists an empty axis-aligned square having these two points on its boundary.

Bonichon \emph{et al.}~\cite{StretchFactorL1Linfty} showed that the $L_\infty$-Delaunay triangulation has a stretch factor of $\cLone = \sqrt{4 + 2\sqrt{2}}$, i.e., there is a path between any two vertices whose length is at most $\cLone$ times their Euclidean distance.
We approximate this path in the $L_\infty$-Delaunay triangulation by showing the existence of a spanning path in the $\theta_4$-graph of $P$ joining the endpoints of every edge in the $L_\infty$-Delaunay triangulation. 
The main ingredient to obtain this approximation  is Lemma~\ref{lemma:One empty quadrant} whose proof is deferred to Section~\ref{sec:One empty quadrant}. Before we can state this lemma, we need a few more definitions.
Given two points $s$ and $t$,
their $L_1$ distance $\Lone{s}{t}$ is the sum of the absolute differences of their $x$- and $y$-coordinates.

Let $S_t(s)$ be the smallest axis-aligned square centered on $t$ that contains $s$.
Let $\ellminus_t$ and $\ellplus_t$ be the lines with slope $-1$ and $+1$ passing through $t$, respectively.

Throughout this paper, we repeatedly use $t$ to denote a \emph{target} point of $P$ that we want to reach via a path in the $\theta_4$-graph. Therefore, we typically omit the reference to $t$ and write $\ellminus, \ellplus$ and $\sq{t}{s}$ when referring to $\ellminus_t, \ellplus_t$ and $S_t(s)$, respectively.

We say that an object is \emph{empty} if its interior contains no point of $P$. An \emph{$s$-$t$-path} is a path with endpoints $s$ and $t$.

\begin{lemma}\label{lemma:One empty quadrant}
Let $s$ and $t$ be two points of $P$ such that $t$ lies in $C_0(s)$.
If the top-right quadrant of $\sq{t}{s}$ is empty and $C_1(s)$ contains no point of $P$ below
$\ell^-$, then there is an $s$-$t$-path in the $\theta_4$-graph of $P$ of length at most $18\cdot \Lone{s}{t}$.
\end{lemma}

Given a path $\varphi$, let $|\varphi|$ denote the sum of the lengths of the edges in $\varphi$.
Using Lemma~\ref{lemma:One empty quadrant}, we obtain the following.

\begin{lemma}\label{lemma:Approximating L_infty edges}
Let $s$ and $t$ be two points of $P$. If the smallest axis-aligned square enclosing $s$ and $t$, that has $t$ as a corner, is empty, then there is an $s$-$t$-path in the $\theta_4$-graph of $P$ of length at most $(\sqrt{2} +  36) \cdot |st|$.
\end{lemma}

\noindent\emph{Proof.}
Assume without loss of generality that $s$ lies in $C_1(t)$. Then, the top-right quadrant of $\sq{t}{s}$ is empty as it coincides with the smallest axis-aligned square enclosing $s$ and $t$ that has $t$ as a corner; see Fig.~\ref{fig:Square Properties}(a).
Recall that $\n{3}{s}$ is the neighbor of $s$ in the $\theta_4$-graph inside the cone $C_3(s)$.
Assume that $\n{3}{s}\neq t$ as otherwise the result follows trivially.
Consequently, $\n{3}{s}$ must lie either in $C_0(t)$ or in $C_2(t)$. Assume without loss of generality that $\n{3}{s}$ lies in the top-left quadrant of $\sq{t}{s}$.
As $\n{3}{s}$ lies in the interior of $\sq{t}{s}$, $\sq{t}{\n{3}{s}}\subset \sq{t}{s}$ and hence, the top-right quadrant of $\sq{t}{\n{3}{s}}$ is empty. 
Moreover, $\n{3}{s}$ lies above $\ellminus$ and hence $C_1(\n{3}{s})$ contains no point of $P$ below $\ellminus$.
Therefore, by Lemma~\ref{lemma:One empty quadrant} there is an $\n{3}{s}$-$t$-path $\varphi$ of length at most $18 \cdot \Lone{\n{3}{s}}{t}$. Since $\n{3}{s}$ lies inside $\sq{t}{s}$, $|\n{3}{s}t| \leq \sqrt{2}\cdot|st|$ and hence
$|\varphi|\leq 18 \cdot \Lone{\n{3}{s}}{t}\leq  18\sqrt{2} \cdot |\n{3}{s}t| \leq 18 \sqrt{2} \sqrt{2} \cdot |st| =36\cdot |st|$.
Moreover, the length of edge $s\n{3}{s}$ is at most $\Lone{s}{t}\leq \sqrt{2}\cdot |st|$ since $\n{3}{s}$ must lie above $\ellminus$.
Thus, $s\n{3}{s} \cup \varphi$ is an $s$-$t$-path of length  
$|s\n{3}{s}| + |\varphi| \leq (\sqrt{2}+ 36) \cdot |st|$.\qed\vspace{.2in}

The following observation is depicted in Fig.~\ref{fig:Square Properties}(b).
\begin{observation}\label{obs:Squares}
Let $S$ be an axis-aligned square.
If two points $a$ and $b$ lie on consecutive sides along the boundary of $S$, then there is a square $S_{ab}$ containing the segment $ab$ such that $S_{ab}\subseteq S$ and either $a$ or $b$ lies on a corner of $S_{ab}$.
\end{observation}

\begin{lemma}\label{lemma:Bounding L_infty edges}
Let $ab$ be an edge of the $L_\infty$-Delaunay triangulation of $P$. There is an $a$-$b$-path $\varphi_{ab}$ in the $\theta_4$-graph of $P$ such that $|\varphi_{ab}| \leq (1+ \sqrt{2}) \cdot (\sqrt{2}+ 36)\cdot |ab|$.
\end{lemma}
\begin{proof}
Let $T  = (a,b,c)$ be a triangle in the $L_\infty$-Delaunay triangulation of $P$. By definition of this triangulation, there is an empty square $S$ such that every vertex of $T$ lies on the boundary of $S$.
By the general position assumption, $a,b$ and $c$ must lie on different sides of $S$. 
If $a$ and $b$ lie on consecutive sides of the boundary of $S$, then by Observation~\ref{obs:Squares} and Lemma~\ref{lemma:Approximating L_infty edges} there is a path $\varphi_{ab}$ contained in the $\theta_4$-graph of $P$ such that $|\varphi_{ab}| \leq (\sqrt{2}+36) \cdot |ab|$.

\begin{figure}[t]
\centering
\includegraphics{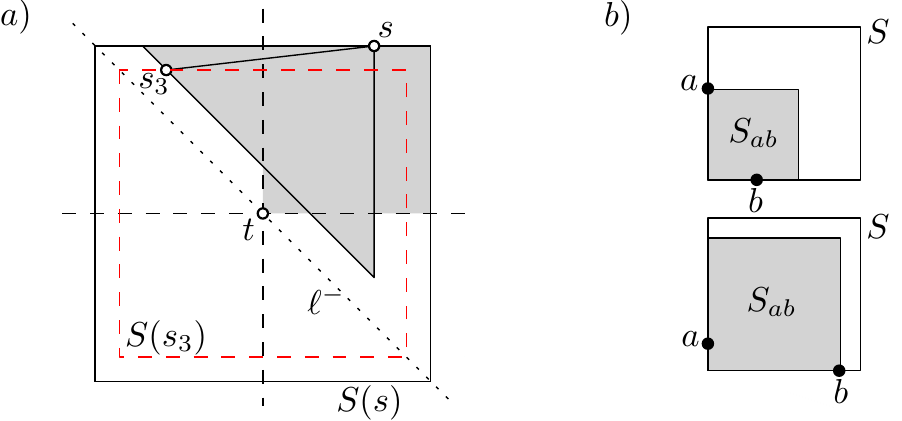}
\caption{\small $a)$ Configuration used in the proof of Lemma ~\ref{lemma:Approximating L_infty edges}, grey areas represent empty regions. $b)$ If $a$ and $b$ lie on consecutive sides of a square $S$, there is a square $S_{ab}$ such that $ab\subset S_{ab}\subseteq S$ and either $a$ or $b$ lies on a corner of $S_{ab}$.}
\label{fig:Square Properties}
\end{figure}

If $a$ and $b$ lie on opposite sides of $S$, then both $ac$ and $cb$ have their endpoints on consecutive sides along the boundary of $S$.
Let $S_{ac}$ be the square contained in $S$ existing as a consequence of Observation~\ref{obs:Squares} when applied on the edge $ac$. Thus, either $a$ or $c$ lies on a corner of $S_{ac}$.
Furthermore, as $S_{ac}$ is contained in $S$, it is also empty.
Consequently, by Lemma~\ref{lemma:Approximating L_infty edges}, there is a $a$-$c$-path $\varphi_{ac}$ such that 
$|\varphi_{ac}|\leq (\sqrt{2}+36)\cdot |ac|$. Analogously, there is a path $\varphi_{cb}$ such that $|\varphi_{cb}|\leq (\sqrt{2}+36)\cdot |cb|$.
Using elementary geometry, it can be shown that since $a$ and $b$ lie on opposite sides of $S$,  $|ac| + |cb| \leq (1+\sqrt{2})\cdot |ab|$.
Therefore, the path $\varphi_{ab} = \varphi_{ac}\cup \varphi_{cb}$ is an $a$-$b$-path such that
$|\varphi_{ab}| \leq (1+ \sqrt{2}) \cdot (\sqrt{2}+36)\cdot |ab|$.
\end{proof}

\begin{theorem}
The $\theta_4$-graph of $P$ is a spanner whose stretch factor is at most $$(1+ \sqrt{2}) \cdot (\sqrt{2}+36)\cdot \sqrt{4 + 2\sqrt{2}}\approx 237$$
\end{theorem}
\begin{proof}
Let $\nu$ be the shortest path joining $s$ with $t$ in the $L_\infty$-Delaunay triangulation of $P$. Bonichon \emph{et al.}~\cite{StretchFactorL1Linfty} proved that the length of $\nu$ is at most $\sqrt{4 + 2\sqrt{2}} \cdot |st|$.
By replacing every edge in $\nu$ with the path in the $\theta_4$-graph of $P$ that exists by Lemma~\ref{lemma:Bounding L_infty edges}, we obtain an $s$-$t$-path of length at most

\quad$(1+ \sqrt{2}) \cdot (\sqrt{2}+36) \cdot |\nu|\leq (1+ \sqrt{2}) \cdot (\sqrt{2}+36)\cdot \sqrt{4 + 2\sqrt{2}} \cdot |st|$
\end{proof}

\section{Light paths}
We introduce some tools that will help us proving Lemma~\ref{lemma:One empty quadrant} in Section~\ref{sec:One empty quadrant}.

Given a point $p$ of $P$, we call edge $p\n{i}{p}$ an \emph{$i$-edge}.
Let $\varphi$ be a path that follows only $0$- and $1$-edges. A $0$-edge $p\n{0}{p}$ of $\varphi$ is \emph{light} if no edge of $\varphi$ crosses the horizontal ray shooting to the right from $p$. We say that $\varphi$ is a \emph{light} path if all its $0$-edges are light.
In this section we show how to bound the length of a light path with respect to the Euclidean distance between its endpoints.

Notice that every $i$-edge is associated with an empty isosceles right triangle. For a point $p$, the empty triangle generated by its $i$-edge is denoted by $\Delta_i(p)$.

\begin{lemma}\label{lemma:0-edges in Light Paths}
Given a light path $\varphi$, every pair of $0$-edges of $\varphi$ has disjoint orthogonal projection on the line defined by the equation $y = -x$.
\end{lemma}
\begin{proof}
Let $s$ and $t$ be the endpoints of $\varphi$.
Let $p\n{0}{p}$ be any $0$-edge of $\varphi$ and let $\nu_{\n{0}{p}}$ be the \emph{diagonal line} extending the hypotenuse of $\Delta_0(p)$, i.e., $\nu_{\n{0}{p}}$ is a line with slope $+1$ passing through $\n{0}{p}$. 
Let $\gamma$ be the path contained in $\varphi$ that joins $\n{0}{p}$ with $t$.
We claim that every point in $\gamma$ lies below $\nu_{\n{0}{p}}$. If this claim is true, the diagonal lines constructed from the empty triangles of every $0$-edge in $\varphi$ split the plane into disjoint slabs, each containing a different $0$-edge of $\varphi$. Thus, their projection on the line defined by the equation $y=-x$ must be disjoint.

To prove that every point in $\gamma$ lies below $\nu_{\n{0}{p}}$, notice that every point in $\gamma$ must lie to the right of $p$ since $\varphi$ is $x$-monotone, and below $p$ since $p\n{0}{p}$ is light, i.e., $\gamma$ is contained in $C_0(p)$. Since $\Delta_0(p)$ is empty, no point of $\gamma$ lies above $\nu_{\n{0}{p}}$ and inside $C_0(p)$ yielding our claim.
\end{proof}

Given a point $w$ of $P$, we say that a point $p$ of $P$ is \emph{$w$-protected} 
if $C_1(p)$ contains no point of $P$ below or on $\ellminus_w$, recall that $\ellminus_w$ is the line with slope $-1$ passing through $w$. 
In other words, a point $p$ is $w$-protected if either $C_1(p)$ is empty or $\n{1}{p}$ lies above $\ellminus_w$. Moreover, every point lying above $\ellminus_w$ is $w$-protected and no point in $C_3(w)$ is $w$-protected.

Given two point $s$ and $t$ such that $s$ lies to the left of $t$, 
we aim to construct a path joining $s$ with $t$ in the $\theta_4$-graph of $P$. 
The role of $t$-protected points will be central in this construction. However, as a first step, we relax our goal and prove instead the existence of a light path $\mpath{s}{t}$ going from $s$ towards $t$ that does not necessarily end at $t$.

To construct $\mpath{s}{t}$, start at a point $z=s$ and repeat the following steps until reaching either $t$ or a $t$-protected point $w$ lying to the right of $t$.
\begin{itemize}
\item If $z$ is not $t$-protected, then follow its $1$-edge, i.e., let $z = \n{1}{z}$. 
\item If $z$ is $t$-protected, then follow its $0$-edge, i.e., let $z = \n{0}{z}$. 
\end{itemize}
The pseudocode of this algorithm can be found in Algorithm~\ref{alg:MonotonePath}.

\begin{algorithm}
  \begin{algorithmic}[1]
    \STATE Let $z = s$.
    \STATE Append $s$ to $\mpath{s}{t}$.
    \WHILE{$z\neq t$ and $z$ is not a $t$-protected point lying to the right of $t$}\label{step:While step}
    \STATE \textbf{if} $z$ is $t$-protected \textbf{then} $z = \n{0}{z}$ \textbf{else} $z = \n{1}{z}$\label{step:If statement}
            \STATE Append $z$ to $\mpath{s}{t}$.
    \ENDWHILE
    \RETURN $\mpath{s}{t}$
  \end{algorithmic}
\caption{Given two points $s$ and $t$ of $P$ such that $s$ lies to the left of $t$, algorithm to compute the path $\mpath{s}{t}$}
\label{alg:MonotonePath}
\end{algorithm}

\vspace{-.1in}
\begin{lemma}\label{lemma:Properties of mpath}
Let $s$ and $t$ be two points of $P$ such that $s$ lies to the left of $t$.
Algorithm~\ref{alg:MonotonePath} produces a light $x$-monotone path $\mpath{s}{t}$ joining $s$ with a $t$-protected point $w$ such that either $w=t$ or $w$ lies to the right of $t$. Moreover, every edge on $\mpath{s}{t}$ is contained in $\sq{t}{s}$.
\end{lemma}

\begin{wrapfigure}[13]{r}{0.36\textwidth}
\centering
\includegraphics{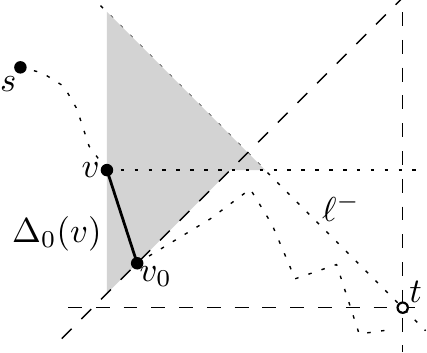}
  \vspace{-.2in}
\caption{If $v$ is a $t$-protected point, then edge $v\n{0}{v}$ is light in any path $\mpath{s}{t}$ that contains it.}
  \label{fig:Light Edge}
\end{wrapfigure}

\noindent\emph{Proof.} By construction, Algorithm~\ref{alg:MonotonePath} finishes only when reaching either $t$ or a $t$-protected point lying to the right of $t$.
Since every edge of $\mpath{s}{t}$ is either a $0$-edge or a $1$-edge traversed from left to right, $\mpath{s}{t}$ is $x$-monotone.

The left endpoint of every $0$-edge in $\mpath{s}{t}$ lies  in $C_2(t)$ as it most be $t$-protected and no $t$-protected point lies in $C_3(t)$.
Thus, if $v\n{0}{v}$ is a $0$-edge, then $v$ lies in $C_2(t)$ and hence, $\n{0}{v}$ lies inside $\sq{t}{s}$ and above $\ellplus$. Otherwise $t$ would lie inside $\Delta_0(v)$. Therefore, every $0$-edge in $\mpath{s}{t}$ is contained in $\sq{t}{s}$.

Every $1$-edge in $\mpath{s}{t}$ has its two endpoints lying below $\ellminus$; otherwise, we followed the $1$-edge of a $t$-protected point which is not allowed by Step~\ref{step:If statement} of Algorithm~\ref{alg:MonotonePath}. Thus, every $1$-edge in $\mpath{s}{t}$ lies below $\ellminus$ and to the right of $s$. As $1$-edges are traversed from bottom to top and the $0$-edges of $\mpath{s}{t}$ are enclosed by $\sq{t}{s}$, every $1$-edge in $\mpath{s}{t}$ is contained in $\sq{t}{s}$.

Let $v\n{0}{v}$ be any $0$-edge of $\mpath{s}{t}$. 
Since we followed the $0$-edge of $v$, we know that $v$ is $t$-protected and hence no point of $P$ lies in $C_1(v)$ and below $\ellminus$. 
As every $1$-edge has its two endpoints lying below $\ellminus$ and $\mpath{s}{t}$ is $x$-monotone, no $1$-edge in $\mpath{s}{t}$ can have an endpoint in $C_1(v)$.
In addition, every $0$-edge of $\mpath{s}{t}$ joins its left endpoint with a point below it. Thus, no $0$-edge of $\mpath{s}{t}$ can cross the ray shooting to the right from $v$. Consequently, $v\n{0}{v}$ is light and hence $\mpath{s}{t}$ is a light path; see Fig~\ref{fig:Light Edge}.
\qed\vspace{.2in}

Given two points $p$ and $q$, let $\xdist{p}{q}$ and $\ydist{p}{q}$ be the absolute differences between their $x$- and $y$-coordinates, respectively, i.e., $\Lone{p}{q} = \xdist{p}{q} + \ydist{p}{q}$.

\begin{lemma}\label{lemma:Light spanners for t-protected points}
Let $s$ and $t$ be two points of $P$ such that $s$ lies to the left of $t$.
If $s$ is $t$-protected, then $|\mpath{s}{t}| \leq 3 \cdot \Lone{s}{t}$.
\end{lemma}
\begin{proof}
To bound the length of $\mpath{s}{t}$, we bound the length of its $0$-edges and the length of its $1$-edges separately.
Let $Z$ be the set of all $0$-edges in $\mpath{s}{t}$ and consider their orthogonal projection on $\ellminus$.
By Lemma~\ref{lemma:0-edges in Light Paths}, all these projections are disjoint. 
Moreover, the length of every $0$-edge in $Z$ is at most $\sqrt{2}$ times the length of its projection. 
Let $s_\perp$ be the orthogonal projection of $s$ on $\ellminus$ and let $\delta$ be the segment joining $s_\perp$ with  $t$. 
Since $s$ is $t$-protected and $\mpath{s}{t}$ is $x$-monotone, the orthogonal projection of every $0$-edge of $Z$ on $\ellminus$ is contained in $\delta$ and hence $\sum_{e\in Z}|e| \leq \sqrt{2}\cdot|\delta|$.
Since $|\delta|= \Lone{s}{t}/\sqrt{2}$ as depicted in Fig.~\ref{fig:Bounding Z and O}(a), we conclude that $\sum_{e\in Z} |e| \leq \Lone{s}{t}$.

Let $O$ be the set of all $1$-edges in $\mpath{s}{t}$ and let $\eta$ be the horizontal line passing through $t$. Since $\mpath{s}{t}$ is $x$-monotone, the orthogonal projections of all edges in $O$ on $\eta$ are disjoint.
Let $\gamma_0, \ldots, \gamma_k$ be the connected components induced by $O$, i.e., the set of maximal connected paths that can be formed by the $1$-edges in $O$; see Fig.~\ref{fig:Bounding Z and O}(b). We claim that the slope of the line joining the two endpoints $p^i, q^i$ of every $\gamma_i$ is smaller than $1$.
If this claim is true, the length of every $\gamma_i$ is bounded by $\xdist{p^i}{q^i} + \ydist{p^i}{q^i}\leq 2\cdot\xdist{p^i}{q^i}$ as each $\gamma_i$ is $x$- and $y$-monotone.

To prove that the slope between $p^i$ and $q^i$ is smaller than $1$, let $v\n{0}{v}$ be the $0$-edge of $\mpath{s}{t}$ such that $\n{0}{v} = p^i$. Since $v\n{0}{v}$ is in $\mpath{s}{t}$, $v$ is $t$-protected by Step~\ref{step:If statement} of Algorithm~\ref{alg:MonotonePath} and hence, as $\Delta_0(v)$ is empty, $q^i$ must lie below the line with slope $+1$ passing through $p^i$ yielding our claim.

Let $\omega$ be the segment obtained by shooting a ray from $t$ to the left until hitting the boundary of $\sq{t}{s}$.
We bound the length of all edges in $O$ using the length of $\omega$. 
Notice that the orthogonal projection of every $\gamma_i$ on $\eta$ is contained in $\omega$, except maybe for $\gamma_k$ whose right endpoint $q^k$ could lie below and to the right of $t$. Two cases arise: If the projection of $\gamma_k$ on $\eta$ is contained in $\omega$, then $\sum_{i=0}^k |\gamma_i| \leq  \sum_{i=0}^{k}  2\cdot\xdist{p^i}{q^i}   \leq 2 \cdot|\omega|$.
Otherwise, since $q_k$  is $t$-protected, $q_k$ lies below $\ellminus$ and hence $\Lone{p^k}{q^k}\leq \Lone{p^k}{t}$.
Moreover, $p^k$ must lie above $\ellplus$ as $p^k$ is reached by a $0$-edge coming from above $\eta$, i.e., $\ydist{p^k}{t}<\xdist{p^k}{t}$. Therefore, $$|\gamma_k| \leq \Lone{p^k}{q^k} \leq \Lone{p^k}{t} = \xdist{p^k}{t} + \ydist{p^k}{t} \leq 2\cdot \xdist{p^k}{t}$$
Consequently, $\sum_{i=0}^k |\gamma_i| \leq 2\cdot \xdist{p^k}{t} + \sum_{i=0}^{k-1}  2\cdot\xdist{p^i}{q^i}   \leq 2 \cdot|\omega|$.
Since $|\omega|\leq \Lone{s}{t}$, we get that $\sum_{e\in O}|e| = \sum_{i=0}^k |\gamma_i| \leq 2\cdot \Lone{s}{t}$.
Thus, $\mpath{s}{t}$ is a light path of length at most $\sum_{e\in O}|e| + \sum_{e\in Z}|e| \leq 3\cdot \Lone{s}{t}$.
\end{proof}

\begin{figure}[t]
\centering
\includegraphics{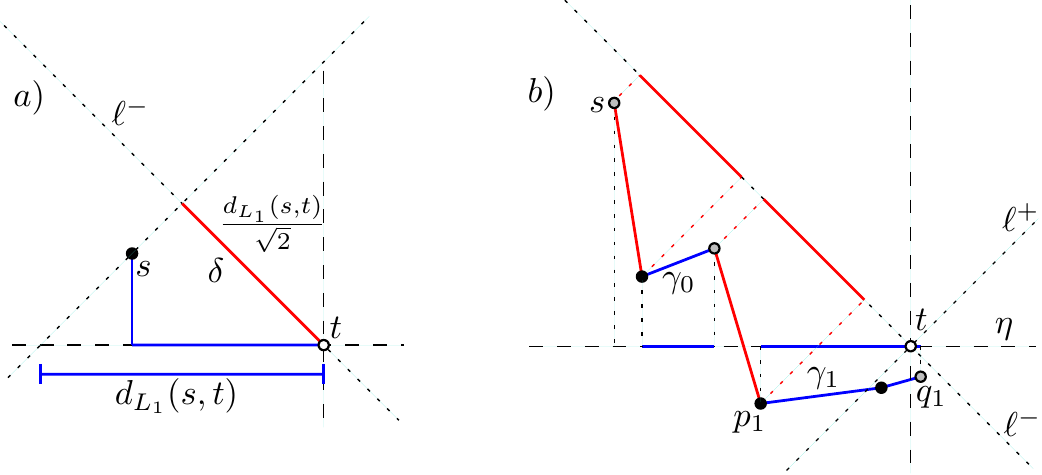}
\caption{\small $a)$ The segment $\delta$ having length $\Lone{s}{t}/\sqrt{2}$.
$b)$ The $0$-edges of $\mpath{s}{t}$ have disjoint projections on $\ellminus$ and the $1$-edges have disjoint projections on the horizontal line passing through $t$. The slope between the endpoints of the maximal paths $\gamma_0$ and $\gamma_1$ is less than 1. }
\label{fig:Bounding Z and O}
\end{figure}

By the construction of the light path in Algorithm~\ref{alg:MonotonePath}, we observe the following.
\begin{lemma}\label{lemma:Endpoint of mpath}
Let $s$ and $t$ be two points of $P$ such that $s$ lies to the left of $t$. If the right endpoint $w$ of $\mpath{s}{t}$ is not equal to $t$, then $w$ lies either above $\ellplus$ if $w\in C_1(t)$, or below $\ellminus$ if $w\in C_0(t)$.
\end{lemma}
\begin{proof}
If $w$ lies in $C_1(t)$, then by Step~\ref{step:If statement} of Algorithm~\ref{alg:MonotonePath}, $w$ was reached by a $0$-edge $pw$ such that $p$ is a $t$-protected point lying above and to the left of $t$. As $\Delta_0(p)$ is empty, $t$ lies below the hypotenuse of $\Delta_0(p)$ and hence $w$ lies above $\ellplus$.

Assume that $w$ lies in $C_0(t)$.
Notice that $w$ is the only $t$-protected point of $\mpath{s}{t}$ that lies to the right of $t$; otherwise, Algorithm~\ref{alg:MonotonePath} finishes before reaching $w$.
By Step~\ref{step:If statement} of Algorithm~\ref{alg:MonotonePath}, every $0$-edge of $\mpath{s}{t}$ needs to have a $t$-protected left endpoint. Moreover, every $t$-protected point of $\mpath{s}{t}$, other that $w$, lies above and to the left of $t$. 
Therefore, $w$ is not reached by a $0$-edge of $\mpath{s}{t}$, i.e., $w$ must be the right endpoint of a $1$-edge $pw$ of $\mpath{s}{t}$.
Notice that $w$ cannot lie above $\ellminus$ since otherwise $p$ is $t$-protected and hence Algorithm~\ref{alg:MonotonePath} finishes before reaching $w$ yielding a contradiction. Thus, $w$ lies below $\ellminus$.
\end{proof}

\section{One empty quadrant}\label{sec:One empty quadrant}

In this section, we provide the proof of Lemma~\ref{lemma:One empty quadrant}. Before stepping into the proof, we need one last definition.
Given a point $p$ of $P$, the \emph{$\max_1$-path} of $p$ is the longest path having $p$ as an endpoint that consists only of $1$-edges and contains the edge $p\n{1}{p}$. We restate Lemma~\ref{lemma:One empty quadrant} using the notions of $t$-protected and $s$-$t$-path.
\\

\noindent \textbf{Lemma~\ref{lemma:One empty quadrant}.}
\emph{Let $s$ and $t$ be two points of $P$ such that $t$ lies in $C_0(s)$.
If the top-right quadrant of $\sq{t}{s}$ is empty and $s$ is $t$-protected, then there is an $s$-$t$-path in the $\theta_4$-graph of $P$ of length at most $18 \cdot\Lone{s}{t}$.}\\

\noindent\emph{Proof.}
Since $s$ is $t$-protected, no point of $P$ lies above $s$, to the right of $s$ and below $\ellminus$; see the dark-shaded region in Fig.~\ref{fig:Base case}.
Let $R$ be the smallest axis-aligned rectangle enclosing $s$ and $t$ and let $k$ be the number of $t$-protected points inside $R$, by the general position assumption, these points are strictly contained in $R$.
We prove the lemma by induction on $k$.

\begin{wrapfigure}[13]{r}{0.38\textwidth}
\centering
\includegraphics{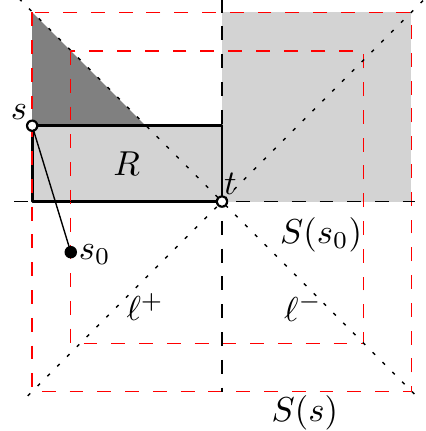}
    \vspace{-.1in}
     \caption{Base case.}
     \label{fig:Base case}
\end{wrapfigure}

\textbf{Base case: } Assume that $R$ contains no $t$-protected point, i.e., $k=0$. We claim that $R$ must be empty and we prove it by contradiction. Let $q$ be a point in $R$ and note that $q$ cannot lie above $\ellminus$ as it would be $t$-protected yielding a contradiction. 
If $q$ lies below $\ellminus$, we can follow the $\max_1$-path from $q$ until reaching a $t$-protected point $p$ lying below $\ellminus$. Since $s$ is $t$-protected, $p$ must lie inside $R$ which is also a contradiction. Thus, $R$ must be empty.

Assume that $\n{0}{s}\neq t$ since otherwise the result is trivial.
As $R$ is empty and $\n{0}{s}\neq t$, $\n{0}{s}$ lies below $t$ and above $\ellplus$. Moreover, no point of $P$ lies above $t$, below $\ellminus$ and inside $\sq{t}{\n{0}{s}}$ since $s$ is $t$-protected. 
Thus, if we think of the set of points $P$ rotated 90 degrees clockwise around $t$, Lemma~\ref{lemma:Light spanners for t-protected points} and Lemma~\ref{lemma:Endpoint of mpath} guarantee the existence of an $\n{0}{s}$-$t$-path $\gamma$ of length at most $3\cdot \Lone{\n{0}{s}}{t}$. 
Since $\n{0}{s}$ lies above $\ellplus$, $\Lone{s}{\n{0}{s}}\leq \Lone{s}{t}$. Furthermore, $\Lone{\n{0}{s}}{t}\leq 2\cdot \Lone{s}{t}$ as $\n{0}{s}$ lies inside $\sq{t}{s}$. Thus, by joining $s\n{0}{s}$ with $\gamma$, we obtain an $s$-$t$-path of length at most $7\cdot \Lone{s}{t}$.


\textbf{Inductive step:} We aim to show the existence of a path $\gamma$ joining $s$ with a $t$-protected point $w\in R$ such that the length of $\gamma$ is at most $18\cdot \Lone{s}{w}$. 
If this is true, we can merge $\gamma$ with the $w$-$t$-path $\varphi$ existing by the induction hypothesis to obtain the desired $s$-$t$-path with length at most $18 \cdot\Lone{s}{t}$. We analyze two cases depending on the position of $\n{0}{s}$ with respect to $R$.

\begin{wrapfigure}[11]{r}{0.35\textwidth}
\centering
\includegraphics{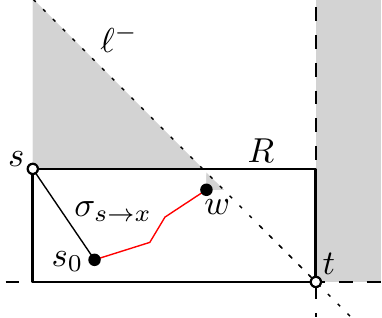}
\vspace{-.1in}
     \caption{Case 1.}
     \label{fig:Case 1}
\end{wrapfigure}

\textbf{Case 1. } Assume that $\n{0}{s}$ lies inside $R$.
If $\n{0}{s}$ lies above $\ellminus$, then $\n{0}{s}$ is $t$-protected and hence we are done after applying our induction hypothesis on $s_0$.
If $\n{0}{s}$ lies below $\ellminus$, then we can follow its $\max_1$-path to reach a $t$-protected point $w$ that must lie inside $R$ as $s$ is $t$-protected.
By running Algorithm~\ref{alg:MonotonePath} on $s$ and $w$, we obtain a path $\mpath{s}{w}$ that goes through the edge $s\n{0}{s}$ and then follows the $\max_1$-path of $\n{0}{s}$ until reaching $w$; see Fig.~\ref{fig:Case 1}. 

Since $s$ is $t$-protected and $w$ lies below $\ellminus$, $s$ is also $w$-protected. Therefore, Lemma~\ref{lemma:Light spanners for t-protected points} guarantees that $|\mpath{s}{w}| \leq 3\cdot\Lone{s}{w}$.
By induction hypothesis on $w$, there is a $w$-$t$-path $\varphi$ such that $|\varphi|\leq 18\cdot\Lone{w}{t}$. 
As $w$ lies in $R$, by joining $\mpath{s}{w}$ with $\varphi$ we obtain the desired $s$-$t$-path of length at most $18 \cdot\Lone{s}{t}$.

\textbf{Case 2.} Assume that $\n{0}{s}$ does not lie in $R$. This implies that $\n{0}{s}$ lies below $t$. Assume also that $\mpath{s}{t}$ does not reach $t$; otherwise we are done since $|\mpath{s}{t}|\leq 3\cdot\Lone{s}{t}$.
Thus, as the top-right quadrant of $\sq{t}{s}$ is empty, $\mpath{s}{t}$ ends at a $t$-protected point $z$ lying in the bottom-right quadrant of $\sq{t}{s}$. 
We consider two sub-cases depending on whether $\mpath{s}{t}$ contains a point inside $R$ or not. 

\begin{wrapfigure}[12]{r}{0.45\textwidth}
\centering
  \vspace{-.1in}
\includegraphics{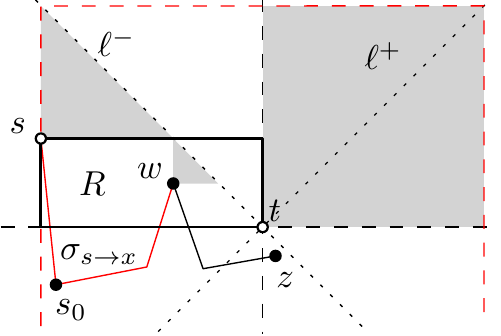}
\vspace{-.0in}
\caption{Case 2.1.}
\label{fig:Case2.1}
\end{wrapfigure}

\textbf{Case 2.1.} If $\mpath{s}{t}$ contains a point inside $R$, let $w$ be the first $t$-protected point of $\mpath{s}{t}$ after $s$ and note that $w$ also lies inside $R$ since $s$ is $t$-protected.
Notice that the part of $\mpath{s}{t}$ going from $s$ to $w$ is in fact equal to $\mpath{s}{w}$ since $w$ lies above $t$ and only $1$-edges were followed after $\n{0}{s}$ by Step~\ref{step:If statement} of Algorithm~\ref{alg:MonotonePath}; see Fig.~\ref{fig:Case2.1}.
Thus, as $s$ is also $w$-protected, the length of $\mpath{s}{w}$ is bounded by $3\cdot\Lone{s}{w}$  by Lemma~\ref{lemma:Light spanners for t-protected points}. Hence, we can apply the induction hypothesis on $w$ as before and obtain the desired $s$-$t$-path.

\textbf{Case 2.2.} If $\mpath{s}{t}$ does not contain a point inside $R$, then $\mpath{s}{t}$ follows only $1$-edges from $\n{0}{s}$ until reaching $z$ in the bottom-right quadrant of $\sq{t}{s}$; see Fig.~\ref{fig:Induction step}(a).
Let $P^*$ be the set of points obtained by reflecting $P$ on the line $\ellplus$. Since $z$ remains $t$-protected after the reflection,
we can use Algorithm~\ref{alg:MonotonePath} to produce a path $\mpath{z}{t}^*$ in the $\theta_4$-graph of $P^*$.
Let $\gamma_{z\to t}$ be the path in the $\theta_4$ graph of $P$ obtain by reflecting $\mpath{z}{t}^*$ on $\ellplus$.
Note that $\gamma_{z\to t}$ ends at a point $w$ such that $w$ is either equal to $t$ or $w$ lies in the top-left quadrant of $\sq{t}{s}$ since the top-right quadrant of $\sq{t}{s}$ is empty.
Since $z$ lies inside $\sq{t}{s}$, $\Lone{z}{t}\leq 2\cdot \Lone{s}{t}$. Therefore, by Lemma~\ref{lemma:Light spanners for t-protected points}, the length of $\mpath{s}{t} \cup \gamma_{z\to t}$ is given by $$|\mpath{s}{t}| + |\gamma_{z\to t}| \leq 3\cdot\Lone{s}{t} + 3\cdot \Lone{z}{t} \leq  3\cdot\Lone{s}{t}+6\cdot\Lone{s}{t} = 9\cdot\Lone{s}{t}.$$

Two cases arise: If $\gamma_{z\to t}$ reaches $t$ ($w=t$), then we are done since $\mpath{s}{t} \cup \gamma_{z\to t}$ joins $s$ with $t$ through $z$. 

If $\gamma_{z\to t}$ does not reach $t$ ($w\neq t$), then $w$ lies below $\ellminus$ by Lemma~\ref{lemma:Endpoint of mpath} applied on $\mpath{z}{t}^*$. Moreover, as $s$ is $t$-protected, no point in $C_1(s)$ can be reached by $\gamma_{z\to t}$ and hence $w$ must lie inside $R$.
We claim that $\Lone{s}{t}\leq 2\cdot \Lone{s}{w}$. If this claim is true, $|\mpath{s}{t} \cup \gamma_{z\to t}| \leq 9\cdot\Lone{s}{t} \leq 18\cdot \Lone{s}{w}$.
Furthermore, by the induction hypothesis, there is a path $\varphi$ joining $w$ with $t$ of length at most $18\cdot\Lone{w}{t}$. 
Consequently, by joining $\mpath{s}{t}, \gamma_{z\to t}$ and $\varphi$, we obtain an $s$-$t$-path of length at most 
$18\cdot\Lone{s}{w} + 18\cdot\Lone{w}{t}  =  18\cdot\Lone{s}{t}$.

\begin{figure}[h!]
\centering
\includegraphics{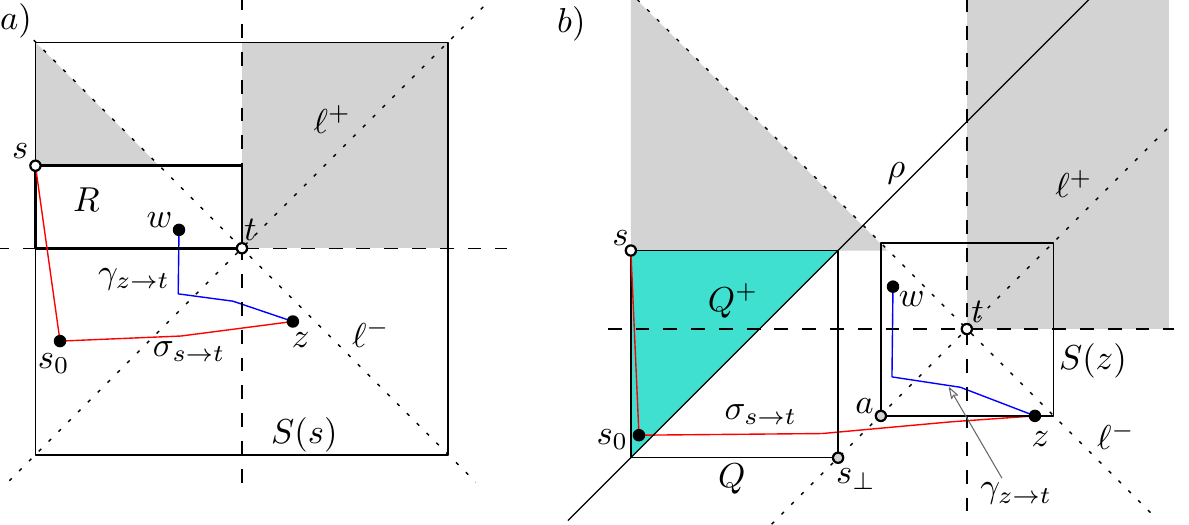}
\caption{\small $a)$ Case 2.2 in the proof of Lemma~\ref{lemma:One empty quadrant}, path $\mpath{s}{t}$ has no point inside $R$ and reaches a point $z$ lying in the bottom-right quadrant of $\sq{t}{s}$.
$b)$ The inductive argument proving that the point $w$, reached after taking the path $\gamma_{z\to t}$, lies outside of the triangle $Q^+$ containing all the points above $\rho$ and below $s$. As $s$ is $t$-protected, the region above $s$ and below $\rho$ is empty.}
\label{fig:Induction step}
\end{figure}

To prove that $\Lone{s}{t}\leq 2\cdot \Lone{s}{w}$, let $s_\perp$ be the orthogonal projection of $s$ on $\ellplus$.
Let $\rho$ be the perpendicular bisector of the segment $ss_\perp$ and notice that for every point $y$ in $C_0(s)$, $\Lone{s}{t} \leq 2\cdot \Lone{s}{y}$ if and only if $y$ lies below $\rho$.

Let $Q$ be the minimum axis-aligned square containing $s$ and $s_\perp$. Note that $\rho$ splits $Q$ into two equal triangles $Q^+$ and $Q^-$ as one diagonal of $Q$ is contained in $\rho$. Assume that $Q^+$ is the triangle that lies above $\rho$.
Notice that all points lying in $C_0(s)$ and above $\rho$ are contained in $Q^+$; see Fig.~\ref{fig:Induction step}(b). We prove that $w$ lies outside of $Q^+$ and hence, that $w$ must lie below $\rho$.

If $\n{0}{s}$ lies below $\rho$, then the empty triangle $\Delta_0(s)$ contains $Q^+$ forcing $w$ to lie below $\rho$.
Assume that $\n{0}{s}$ lies above $\rho$. In this case, $z$ lies above $\n{0}{s}$ as we only followed $1$-edges to reach $z$ in the construction of $\mpath{s}{t}$ by Step~\ref{step:If statement} of Algorithm~\ref{alg:MonotonePath}.
Let $a$ be the intersection of $\ellplus$ and the ray shooting to the left from $z$. Notice that $w$ must lie to the right of $a$ as the path $\gamma_{z\to t}$ is contained in the square $\sq{t}{z}$ and $a$ is one of its corners. 
As $z$ lies above $\n{0}{s}$ and $\n{0}{s}$ lies above $s_\perp$, we conclude that $a$ is above $s_\perp$ and both lie on $\ellplus$. Therefore, $a$ lies to the right of $s_\perp$, implying that $w$ lies to the right of $s_\perp$ and hence outside of $Q^+$.
As we proved that $w$ lies below $\rho$, we conclude that $\Lone{s}{t}\leq 2\cdot \Lone{s}{w}$.\qed

\section{Lower Bound}\label{sec:LowerBound}
In this section we show how to construct a lower bound of 7 for the $\theta_4$-graph. We start with two vertices $u$ and $w$
such that $w$ lies in $C_2(u)$ and the difference of their $x$-coordinates is arbitrarily small.
To construct the lower bound, we repeatedly replace a single edge of the shortest $u$-$w$-path by placing vertices in the corners of the empty triangle(s) associated with that edge. The final graph is shown in Fig.~\ref{fig:LowerBound}. 

We start out by removing the edge between $u$ and $w$ by placing two vertices, one inside $\Delta_2(u)$ and one inside $\Delta_0(w)$, both arbitrarily close to the corner that does not contain $u$ nor $w$. 
Let $v_1$ be the vertex placed in $\Delta_2(u)$. 
Placing $v_1$ and the other vertex in $\Delta_0(w)$ removed edge $u w$, but created two new shortest paths, $u v_1 w$ being one of them. Hence, our next step is to extend this path. 

We remove edge $v_1 w$ (and its equivalent in the other path) by placing a vertex arbitrarily close to the corner of 
$\Delta_1(v_1)$ and $\Delta_3(w)$ that is farthest from $u$. Let $v_2$ be the vertex placed inside $\Delta_1(v_1)$. Hence, edge $v_1 w$ is replaced by the path $v_1 v_2 w$. 

Next, we extend the path again by removing edge $v_2 w$ (and its equivalent edge in the other paths). Like before, we place a vertex arbitrarily close to the corner of 
$\Delta_0(v_2)$ and $\Delta_2(w)$ that is farthest from $u$. Let $v_3$ be the vertex placed in $\Delta_0(v_2)$. Hence, edge $v_2 w$ is replaced by $v_2 v_3 w$. 

Finally, we replace edge $v_3 w$ (and its equivalent edge in the other paths). For all paths for which this edge lies on the outer face, we place a vertex in the corner of the two empty triangles defining that edge. However, for edge $v_3 w$ which does not lie on the outer face, we place a single vertex $v_4$ in the intersection of 
$\Delta_3(v_3)$ and $\Delta_1(w)$. In this way, edge $v_3 w$ is replaced by $v_3 v_4 w$. When placing $v_4$, we need to ensure that no edge $u v_4$ is added as this would created a shortcut. This is easily achieved by placing $v_4$ such that it is closer to $v_3$ than to $w$. The resulting graph is shown in Fig.~\ref{fig:LowerBound}. 

\begin{figure}[ht]
\centering
  \includegraphics{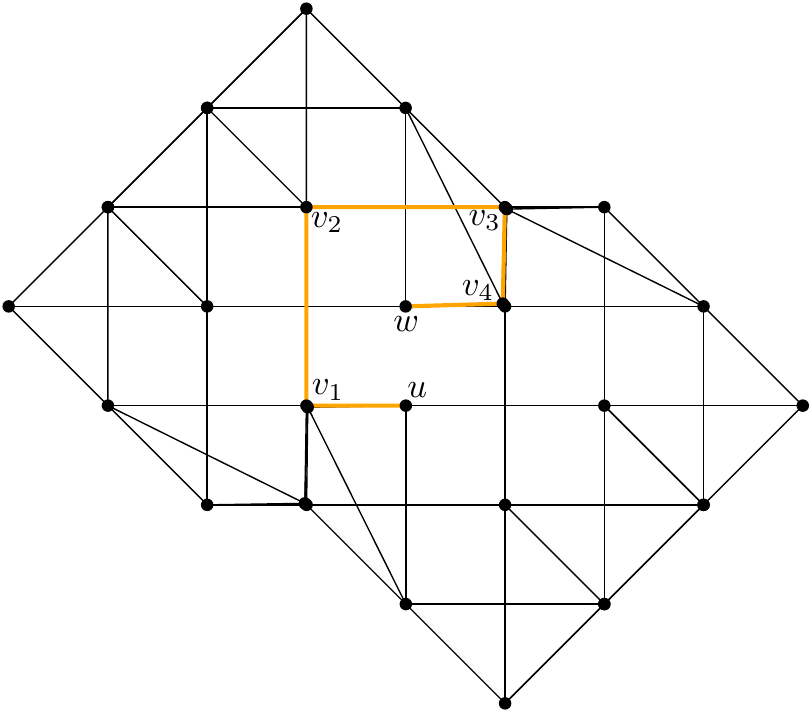}
  \caption{A lower bound for the $\theta_4$-graph. One of the shortest paths from $u$ to $w$ goes via $v_1$, $v_2$, $v_3$, and $v_4$.} 
  \label{fig:LowerBound}
\end{figure}

\begin{lemma}
  The stretch factor of the $\theta_4$-graph is at least $7$. 
\end{lemma}
\begin{proof}
  We look at path $u v_1 v_2 v_3 v_4 w$ from Fig.~\ref{fig:LowerBound}. Edges $u v_1$, $v_3 v_4$, and $v_4 w$ have length $|u w| - \varepsilon$ and edges $v_1 v_2$ and $v_2 v_3$ have length $2 \cdot |u w| - \varepsilon$, where $\varepsilon$ is positive and arbitrarily close to 0. Hence the stretch factor of this path is arbitrarily close to 7. 
\end{proof}

{\small
\bibliographystyle{abbrv}
\bibliography{Theta4}}

\end{document}